\documentclass[11pt]{llncs}
\usepackage{amssymb,amsmath,color}

\usepackage[margin=1.75in]{geometry}

\def\bx{\vec x}
\def\by{\vec y}
\def\bp{\vec p}
\def\bz{\vec z}

\numberwithin{theorem}{section}
\numberwithin{lemma}{section}
\numberwithin{corollary}{section}

\title{The Combinatorial World (of Auctions) According to GARP}
\author{Shant Boodaghians\inst{1} \and Adrian Vetta\inst{2}}
\institute{Department of Mathematics and Statistics, McGill University
\email{shant.boodaghians@mail.mcgill.ca}
\and Department of Mathematics and Statistics, and School of Computer 
Science, McGill University
\email{vetta@math.mcgill.ca}}

\begin{document}

\maketitle

\addtocounter{footnote}{-2}

\begin{abstract}
Revealed preference techniques are used to test whether a data set is 
compatible with rational behaviour. 
They are also incorporated as constraints in mechanism design to encourage 
truthful behaviour in applications such as combinatorial auctions. 
In the auction setting, we present an efficient combinatorial algorithm to 
find a virtual valuation function with the optimal (additive) rationality 
guarantee. 
Moreover, we show that there exists such a valuation function that both is individually rational and is minimum 
(that is, it is component-wise dominated by any other individually rational, 
virtual valuation function that approximately fits the data). 
Similarly, given upper bound constraints on the valuation function, 
we show how to fit the maximum virtual valuation function with the optimal 
additive rationality guarantee. 
In practice, revealed preference bidding constraints are very demanding. 
We explain how approximate rationality can be used to create relaxed revealed 
preference constraints in an auction. 
We then show how combinatorial methods can be used to implement these relaxed 
constraints. 
Worst/best-case welfare guarantees that result from the use of such mechanisms 
can be quantified via the minimum/maximum virtual valuation function.
\end{abstract}

\section{Introduction}

Underlying the theory of consumer demand is a standard rationality assumption: 
given a set of items with price vector $\bp$, a consumer will demand the 
bundle $\bx$ of maximum utility whose cost is at most her budget $B$. 
Of fundamental import, therefore, is whether or not the decision making \
behaviour of a real consumer is consistent with the maximization of a utility function. 
Samuelson~\cite{Sam38,Sam48} introduced {\em revealed preference} to provide 
a theoretical framework within which to analyse this question. 
Furthermore, this concept now lies at the heart of current empirical work in 
the field; see, for example, Gross~\cite{Gro95} and Varian~\cite{Var07}. 
Specifically, Samuelson~\cite{Sam38} conjectured that the 
{\em weak axiom of revealed preference} ({\sc warp}) was a necessary and 
sufficient condition for {\em integrability} -- 
the ability to construct a utility function which fits observed behaviour. 

However, Houtthakker~\cite{Hou50} proved that the weak axiom was insufficient.
Instead, he presented a {\itshape strong axiom of revealed preference} 
({\sc sarp}) and showed non-constructively that it was necessary and 
sufficient in the case where behaviour is determined via a single-valued 
demand function. 
Afriat~\cite{Afr67} provided an extension to multi-valued demand functions 
-- where ties are allowed -- by showing that the 
{\em generalized axiom of revealed preference} ({\sc garp}) 
is necessary and sufficient for integrability.%
\footnote{
	Afriat~\cite{Afr67} gave several equivalent necessary and sufficient 
	conditions for integrability.
	One of these, {\em cyclical consistency}, is equivalent to {\sc garp} 
	as shown by Varian~\cite{Var82}.
} 
Furthermore, Afriat's approach was constructive 
(producing monotonic, concave, piecewise-linear utility functions) 
and applied to the setting of a finite collection of observational data. 
This rendered his method more suitable for practical use. 

In addition to its prominence in testing for rational behaviour, 
revealed preference has become an important tool in mechanism design. 
A notable area of application is auction design. 
For combinatorial auctions, Ausubel, Cramton and Milgrom~\cite{ACM06} 
proposed bidding activity rules based upon {\sc warp}. 
These rules are now standard in the combinatorial clock auction, 
one of the two prominent auction mechanisms used to sell bandwidth. 
In part, the {\sc warp} bidding rules have proved successful because 
they are extremely difficult to game \cite{Cra13}. 
Harsha et al.~\cite{HBP10} examine {\sc garp}-based bidding rules, 
and Ausubel and Baranov~\cite{AB14} advocate incorporating such constraints 
into bandwidth auctions. 
Based upon Afriat's theorem, these {\sc garp}-based rules imply that there 
always exists a utility function that is compatible with the bidding history. 
This gives the desirable property that a bidder in an auction will always have 
at least one feasible bid 
-- a property that cannot be guaranteed under {\sc warp}.

Revealed preference also plays a key role in motivating the 
{\em generalised second price mechanism} used in adword auctions. 
Indeed, these position auctions have welfare maximizing solutions 
with respect to a revealed preference equilibrium concept; 
see \cite{Var07} and \cite{EOS07}. 

\subsection{Our Results}
Multiple methods have been proposed to approximately measure how consistent a 
data set is with rational behaviour; 
see Gross~\cite{Gro95} for a comparison of a sample of these approaches. 
In this paper, we show how a graphical viewpoint of revealed preference 
can be used to obtain a virtual valuation function that best fits the data set. 
Specifically, we show in Section~\ref{sec:approx-vv} that an individually 
rational virtual valuation function can be obtained such that its 
additive deviation from rationality is exactly the minimum mean length of a 
cycle in a bidding graph. 
This additive guarantee cannot be improved upon. 
Furthermore, we show there exists a unique {\em minimum} valuation function 
from amongst all individually rational virtual valuation functions that 
optimally fit the data. 
Similarly, given a set of upper bound constraints, we show how to find the 
unique {\em maximum} virtual valuation that optimally fits the data, if it exists.

Imposing revealed preference bidding rules can be harsh. 
Indeed, Cramton~\cite{Cra13} states that 
``there are good reasons to simplify and somewhat weaken the revealed 
preference rule". 
These reasons include complexity issues, common value uncertainty, 
the complication of budget constraints, and the fact that a bidder's 
assessment of her valuation function often {\em changes} as the 
auction progresses! 
The concept of approximate rationality, however, naturally induces a 
relaxed form of revealed preference rules. 
We examine such relaxed bidding rules in Section~\ref{sec:mech-design}, 
show how they can be implemented combinatorially, and show how 
to construct the minimal and maximal valuation functions which fit the data, 
which may be useful for quantifying worst/best-case welfare guarantees.

\section{Revealed Preference}\label{sec:rp}
In this section, we first review revealed preference. 
We then examine its use in auction design and describe how to formulate it 
in terms of a bidding graph.

\subsection{Revealed Preference with Budgets}\label{sec:standard-rp}
The standard revealed preference model instigated by Samuelson \cite{Sam38} 
is as follows. 
We are given a set of observations 
\[
	\{(B_1,\bp_1, \bx_1),\,(B_2,\bp_2, \bx_2),\,\dotsc,(B_T,\bp_T, \bx_T)\}
	\enspace.
\] 
At time $t$, $1\le t\leq T$, the set of items has a price vector $\bp_t$ and 
the consumer chooses to spend her budget $B_t$ on the bundle $\bx_t$.%
\footnote{
	It is not necessary to present the model in terms of ``time". 
	We do so because this best accords with the combinatorial auction application.
} 
We say that bundle $\bx_t$ is (directly) revealed preferred to bundle $\by$, 
denoted $\bx_t\succeq \by$, if $\by$ was affordable when $\bx_t$ was purchased. 
We say that bundle $\bx_t$ is strictly revealed preferred to bundle~$\by$, 
denoted $\bx_t\succ \by$, if $\by$ was (strictly) cheaper than $\bx_t$ 
when $\bx_t$ was purchased. 
This gives {\itshape revealed preference}~(\ref{rp}) and 
{\em strict revealed preference}~(\ref{srp}):
\begin{eqnarray}
	\bp_t\cdot\by \le  \bp_t\cdot\bx_t  \ \ \Rightarrow \ \ \bx_t\succeq \by 
		\label{rp} \\
	\bp_t\cdot\by < \bp_t\cdot\bx_t  \ \ \Rightarrow \ \  \bx_t\succ \by 
		\label{srp}
\end{eqnarray}

Furthermore, a basic assumption is that the consumer optimises a locally 
non-satiated utility function.%
\footnote{
	Local non-satiation states that for any bundle $\bx$ there is a more 
	preferred bundle arbitrarily close to $\bx$. 
	A monotonic utility function is locally non-satiated, but the converse 
	need not hold.
} 
Consequently, at time $t$ she will spend her entire budget, $i.e.$, 
$\bp_t\cdot\bx_t=B_t$. 
In the absence of ties, preference orderings give relations that are 
anti-symmetric and transitive. 
This leads to an axiomatic approach to revealed preference formulated in 
terms of {\sc warp} and {\sc sarp} by Houthakker~\cite{Hou50}. 
The {\em weak axiom of revealed preference} ({\sc warp}) states that the 
relation should be asymmetric, $i.e.$ 
${\bx\succeq\by\Rightarrow\by\not \succeq\bx}$. 
Its transitive closure, the {\em strong axiom of revealed preference} 
({\sc sarp}) states that the relation should be acyclic. 
Our interest lies in the general case where ties are allowed. 
This produces what we dub the 
{\itshape $k$-th Axiom of Revealed Preference} ({\sc karp}): 
Given a fixed integer $k$ and any $\kappa\le k$
\begin{equation}
	\bx_{t}= \bx_{t_0}
		\succeq 	\bx_{t_1} 
		\succeq 	\cdots 
		\succeq \bx_{t_{\kappa-1}} 
		\succeq \bx_{t_\kappa} = \by 
	\quad
	\Rightarrow
	\quad
	\by \nsucc \bx_{t} \label{krp} \enspace .
\end{equation}

There are two very important special cases of {\sc karp}. 
For $k=1$, this is simply {\sc warp}, 
$i.e.$ $\bx_{t}\succeq \by \, \Rightarrow \, \by\not\succ \bx_{t}$. 
This is just the basic property that for a preference ordering, 
we cannot have that $\by$ is strictly preferred to $\bx_t$ if $\bx_t$ is 
preferred to $\by$. 
On the other hand, suppose we take $k$ to be arbitrarily large 
(or simply larger than the total number of observed bundles). 
Then we have the {\em Generalized Axiom of Revealed Preference} ({\sc garp}), 
the simultaneous application of {\sc karp} for each value of $k$. 
In particular, {\sc garp} encodes the property of transitivity of 
preference relations. 
Specifically, for any $k$, if \[
	\bx_{t}= \bx_{t_0}
		\succeq \bx_{t_1} 
		\succeq \cdots 
		\succeq \bx_{t_{k-1}} 
		\succeq \bx_{t_k}=\by\enspace ,
\] then, by transitivity, $\bx_t \succeq \by$. 
The first axiom of revealed preference then implies that $\by\nsucc \bx_t$.

The underlying importance of {\sc garp} follows from a classical result 
of Afriat~\cite{Afr67}: 
there exists a nonsatiated, monotone, concave utility function 
that rationalizes the data if and only if the data satisfy {\sc garp}. 
Brown and Echenique \cite{BE09} examine the setting of indivisible goods 
and Echenique et al. \cite{EGW11} consider the consequent 
computational implications. 

\subsection{Revealed Preference in Combinatorial Auctions}\label{sec:auction-rp}
As discussed, a major application of revealed preference in mechanism design 
concerns combinatorial auctions. 
Here, there are some important distinctions from the standard revealed 
preference model presented in Section~\ref{sec:standard-rp}. 
First, consumers are assumed to have quasilinear utility functions that are 
linear in money. 
Thus, they seek to maximise profit. 
Second, the standard assumption is that bidders have {\em no} budgetary constraints. 
For example, if profitable opportunities arise that require large 
investments then these can be obtained from perfect capital markets. 
(This assumption is slightly unrealistic; Harsha et al.~\cite{HBP10} 
show how to implement a budgeted revealed preference model for 
combinatorial auctions; see also Section~\ref{sec:other}). 

Third, the observations $(\bp_t, \bx_t)$, for each $1\le t\le T$, 
are typically not purchases but are bids made over a collection of auction rounds. 
When offered a set of prices at time $t$ the consumer bids for bundle $\bx_t$. 

So what would a model of revealed preference be in this combinatorial auction setting? Suppose that at time $t$ we select bundle $\bx_t$ and that at an 
earlier time $s$ we selected bundle $\bx_s$.
Assuming a quasi-linear utility function and no budget constraint, 
we have revealed:
\begin{eqnarray}
	v(\bx_t) -\bp_t\cdot\bx_t &\ \ge\ & v(\bx_s) -\bp_t\cdot\bx_s  
		\label{eq:RP2}\\
	v(\bx_s) -\bp_s\cdot\bx_s &\ \ge\ & v(\bx_t) -\bp_s\cdot\bx_t  
		\label{eq:RP1}
\end{eqnarray}

Summing Inequalities (\ref{eq:RP2}) and (\ref{eq:RP1}) and rearranging gives
\begin{equation}
	\label{eq:RP-auction}
	(\bp_t-\bp_s)\cdot\bx_s \ \ge\  (\bp_t-\bp_s)\cdot\bx_t
\end{equation}
This is the revealed preference condition for combinatorial auctions proposed 
as a bidding activity rule by Ausubel, Crampton and Milgrom~\cite{ACM06}.
The activity rule simply states that, between time $s$ and time $t$, 
the price of bundle $\bx_s$ must have risen by at least as
much as the price of $\bx_s$. 
If condition (\ref{eq:RP-auction}) is not satisfied then the auction 
mechanism will not allow the later bid to be made.

Observe that the bidding rule (\ref{eq:RP-auction}) was derived directly 
from the assumption of utility maximisation. 
This unbudgeted revealed preference auction model can, though, 
also be viewed within the framework of the standard budgeted model of 
revealed preference.
To do this, we assume the bidder has an arbitrarily large budget $B$.
In particular, prices will never be so high that she cannot afford to buy 
every item. 
Second, to model quasilinear utility functions, we treat money as a good.
Specifically, given a bundle of items $\bx = (x_1,\,\dotsc,\, x_n)$ and an 
amount $x_0$ of money we denote by $\hat\bx=  (x_0, x_1,\,\dotsc,\, x_n)$ 
the concatenation of $x_0$ and $\bx$. 
If $\bp=(p_1,\dots,p_n)$ is the price vector for the the non-monetary items, 
then $\hat\bp=(1,p_1,\dots,p_n)$ gives the prices of all items including money.

In this $n+1$ dimensional setting, let us select bundle $\hat\bx_t$ at time $t$. 
As the budget $B$ is arbitrarily large, 
we can certainly afford the bundle ${\bx}_s$ at this time. 
But we may not be able to afford bundle $\hat\bx_s$, 
as then we must also pay for the monetary component at a cost of 
$B-\bp_s\cdot\bx_s$. 
However, we can afford the bundle ${\bx}_s$ plus an amount $B-\bp_t\cdot\bx_s$ 
of money. 
Applying revealed preference to $\{\hat\bx, \hat\bp\}$, 
we have revealed that 
${\hat\bx_t =(B-\bp_t\cdot\bx_t, \bx_t)\succeq (B-\bp_t\cdot\bx_s, \bx_s)}$. 
Hence, by quasilinearity, subtracting the monetary component from both sides, 
we have, 
\begin{equation*}
	(0, \bx_t)
		\, \succeq \, ((B-\bp_t\cdot\bx_s)-(B-\bp_t\cdot\bx_t), \bx_s)
		\, = \,        (\bp_t\cdot\bx_t-\bp_t\cdot\bx_s, \bx_s)
	\enspace .
\end{equation*}
Equivalently, 
\begin{equation}
	v(\bx_t) \ge v(\bx_s)+\bp_t\cdot\bx_t-\bp_t\cdot\bx_s \enspace.
	\label{eq:RP-restate}
\end{equation} 
But Inequality (\ref{eq:RP-restate}) is equivalent to Inequality (\ref{eq:RP2}). 
Inequality (\ref{eq:RP1}) follows symmetrically, and together 
these give the revealed preference bidding rule (\ref{eq:RP-auction}). 
Note that this bidding rule is derived via the direct comparison of two bundles.

We can now extend this bidding rule to incorporate indirect comparisons 
in a similar fashion to the extension from {\sc warp} to {\sc sarp} 
via transitivity. 
This produces a {\sc garp}-based bidding rule. 
Namely, suppose we bid for the money-less bundle $\bx_i$ at time $t_i$, 
for all $0\leq i\leq k$, where $1\leq t_i\leq T$. 
Thus we have revealed that
\begin{align*}
	(0, \bx_i) 
		\, &\succeq \, 
			((B-\bp_i\cdot\bx_{i+1})-(B-\bp_i\cdot\bx_i), \bx_{i+1})\\
		\, &= \,
			(\bp_i\cdot\bx_i-\bp_i\cdot\bx_{i+1}, \bx_{i+1})
\end{align*}
This induces the inequality 
\begin{equation}
	v(\bx_i)-\bp_i\cdot\bx_i \geq v(\bx_{i+1})-\bp_i\cdot\bx_{i+1}
	\enspace . \label{eq:rev}
\end{equation}
Summing (\ref{eq:rev}) over all $i$, we obtain 
\begin{equation*}
	\sum_{i=0}^k \left( 
		v(\bx_i)-\bp_i\cdot\bx_i
	\right)\,\geq\,\sum_{i=0}^k \left( 
		v(\bx_{i+1})-\bp_i\cdot\bx_{i+1}
	\right) \enspace, 
\end{equation*} where the sum in the subscripts are taken modulo $k$. 
Rearranging now gives the combinatorial auction {\sc karp}-based 
bidding activity rule:
\begin{equation}\label{eq:garp}
	(\bp_k-\bp_0)\cdot\bx_0
	\, \geq \,
	\sum_{i=1}^k  (\bp_i-\bp_{i-1})\cdot\bx_i \enspace .
\end{equation}
For $k$ arbitrarily large, this gives the {\sc garp}-based bidding rule. 
In order to qualitatively analyze the consequences of imposing 
{\sc karp}-based activity rules, it is informative to now provide a 
graphical interpretation of the these rules.

\subsection{A Graphical View of Revealed Preference}\label{multigraph}
Given the set of price-bid pairings $\{(\bp_t, \bx_t): 1\le t\le T\}$, 
we create a directed graph $G=(V,A)$, called the {\em bidding graph}, 
to which we will assign arc lengths $\ell$. 
There is a vertex in $V$ for each possible bundle 
-- that is, there are $2^n$ bundles in an $n$-item auction. 
For each observed bid $\bx_t$, $1\le t\le T$, there is an arc $(\bx_t, \by)$ 
for each bundle $\by\in V$. 
In order to define the length $\ell_{\bx_t,\by}$ of an arc $(\bx_t, \by)$, 
note that Inequality~(\ref{eq:RP2}) applied to $\bx_s=\by$ gives
\begin{equation*}
	v(\by)\ \le\  v(\bx_t)+ \bp_t\cdot (\by- \bx_t ) \enspace,
\end{equation*} 
otherwise we would prefer bundle $\by$ at time $t$. For the arc length, we would like to simply set  $\ell_{\bx_t,\by} = \bp_t\cdot (\by - \bx_t )$. Observe, however, that the bundle $\bx_t$ may be chosen in more than one time period. That is, possibly $\bx_t=\bx_{t'}$ for some $t\neq t'$. Therefore the bidding graph is, in fact, a multigraph. It suffices, though, to represent only the most stringent constraints imposed by the bidding behaviour. Thus, we obtain a simple graph by setting 
\begin{equation*}
	\ell_{\bx_t, \by} \, =\, 
	\min_{t'} \,\{ 
		\bp_{t'}\cdot (\by - \bx_t): \bx_{t'} =\bx_t
	\}\enspace .
\end{equation*} 
Now the {\sc warp}-based bidding rule (\ref{eq:RP-auction}) of Ausubel et al. \cite{ACM06} is equivalent to 
\begin{equation*}
	(\bp_t-\bp_s)\cdot\bx_s - (\bp_t-\bp_s)\cdot\bx_t \ge  0\enspace .
\end{equation*} 
\vspace*{-2em}

\noindent However,
\vspace*{-.75em}
\begin{align*}
	&\ell_{\bx_s,\bx_t} + \ell_{\bx_t,\bx_s} \\
	&\quad= \min_{s'} \,\{ 
		\bp_{s'}\cdot (\bx_t- \bx_s): \bx_{s'}=\bx_s
	\} +\min_{t'} \,\{ 
		\bp_{t'}\cdot (\bx_s - \bx_t): \bx_{t'}=\bx_t
	\}\\
 	&\quad\leq \bp_s\cdot (\bx_t - \bx_s ) + \bp_t\cdot (\bx_s - \bx_t ) \\
	&\quad= (\bp_t-\bp_s)\cdot\bx_s - (\bp_t-\bp_s)\cdot\bx_t\enspace .
\end{align*}
It is then easy to see that the bidding constraint (\ref{eq:RP-auction}) is 
violated if and only if the bidding graph contains no negative digons 
(cycles of length two). 
Furthermore, we can interpret {\sc karp} and {\sc garp} is a similar fashion. 
Hence, the $k$-th axiom of revealed preference is equivalent to requiring that 
the bidding graph not contain any negative cycles of cardinality at most $k+1$, 
and {\sc garp} is equivalent to requiring no negative cycles at all. 
Thus, we can formalize the preference axioms in terms of the lengths of 
negative cycles in a directed graph. 
We remark that a cyclic view of revealed preference is briefly outlined 
by Vohra~\cite{Voh11}. 
For us, this cyclic formulation has important consequences in testing 
for the extent of bidding deviations from the axioms. 
We will quantify this exactly in Section \ref{sec:approx-vv}. 
Before doing so, though, we remark that the focus on cycles also has 
important computational consequences. 

First, recall that the bidding graph $G$ contains an exponential number 
of vertices, one for every subset of the items. 
Of course, it is not practical to work with such a graph. 
Observe, however, that a bundle ${\by\notin \{\bx_1, \bx_2\dots, \bx_T\}}$ 
has zero out-degree in $G$. 
Consequently, $\by$ cannot be contained in any cycle. 
Thus, it will suffice to consider only the subgraph induced by the bids 
$\{\bx_1, \bx_2\dots, \bx_T\}$. 
In a combinatorial auction there is typically one bid per time period and the 
number of periods is quite small.%
\footnote{
	For example, in a bandwidth auction there are at most a few hundred rounds.
} 
Hence, the induced subgraph of the bidding graph that we actually need is of a 
very manageable size.

Second, one way to implement a bidding rule is via a mathematical program; 
see, for example, Harsha et al.~\cite{HBP10}. 
The cyclic interpretation of a bidding rule has two major advantages: 
we can test the rule very quickly by searching for negative cycles in a graph. 
For example, we can test for negative cycles of length at most $k+1$ either by 
fast matrix multiplication or directly by looking for shortest paths of length 
$k$ using the Bellman-Ford algorithm in $O(T^3)$ time. 
Another major advantage is that a bidder can interpret the consequence of 
a prospective new bid dynamically by consideration of the bidding graph. 
This is extremely important in practice. 
In contrast, bidding rules that require using an optimization solver as a 
black-box are very opaque to bidders.

\section{Approximate Virtual Valuation Functions}\label{sec:approx-vv}
For combinatorial auctions, Afriat's result that {\sc garp} is necessary 
and sufficient for rationalisability can be reformulated as:
\begin{theorem}\label{thm:garp suffices} 
	A valuation function which rationalises bidding behaviour exists if and 
	only if the bidding graph has no negative cycle.
\end{theorem}
This is a simple corollary of Theorem \ref{thm:mu} below; 
see also \cite{Voh11}. 
From an economic perspective, however, what is most important is not whether 
agents are perfectly rational but 
``whether optimization is a reasonable way to describe some behavior"~\cite{Var90}.%
\footnote{
	Indeed, several schools of thought in the field of bounded rationality 
	argue that people utilize simple (but often effective) heuristics rather 
	than attempt to optimize; see, for example, \cite{GS01}.
} 
It is then important to study the consequences of approximately rational behaviour, 
see, for example, Akerlof and Yellen~\cite{AY85}. 
First, though, is it possible to quantify the degree to which agents are rational?
Gross~\cite{Gro95} examines assorted methods to test the degree of rationality. 
Notable amongst them is the {\em Afriat Efficiency Index}~\cite{Afr67,Var90}. 
Here the condition required to imply a preference is strengthened multiplicatively. 
Specifically, $\bx_t\succeq \by$ only if 
$\bp_t\cdot\by \le \lambda\cdot \bp_t\cdot\bx_t$ where $\lambda<1$. 
We examine this index with respect to the bidding graph in Section~\ref{sec:other}.
For combinatorial auctions, a variant of this constraint was examined 
experimentally by Harsha et al.~\cite{HBP10}.

Here we show how to quantify exactly the degree of rationality present in the 
data via a parameter of the bidding graph.
Moreover, we are able to go beyond multiplicative guarantees and obtain 
stronger additive bounds. 
To wit, we say that $\hat v$ is an 
{\em $\epsilon$-approximate virtual valuation function} if, 
for all $t$ and for any bundle~$\by$, 
\begin{equation*}
	\hat v(\bx_t) -\bp_t\cdot\bx_t 
		\, \geq \, 
		\hat v(\by)- \bp_t\cdot \by -\epsilon
		\enspace .
\end{equation*} 
Note that if $\epsilon=0$, then the bidder is optimizing with respect to a virtual valuation function, $i.e.$ is rational. 
We remark that the term {\em virtual} reflects the fact that $\hat v$ need 
not be the real valuation function (if one exists) of the bidder, 
but if it is then the bidding is termed {\em truthful}.

\subsection{Minimum Mean Cycles and Approximate Virtual Valuations}
We now examine exactly when a bidding strategy is approximately rational. 
It turns out that the key to understanding approximate deviations from 
rationality is the \textit{minimum mean cycle} in the bidding graph. 
Given a cycle $C$ in $G$, its mean length is 
\begin{equation*}
	\mu(C)=\frac{\sum_{a\in C} \ell_a}{|C|} \enspace.
\end{equation*} 
We denote by $\mu(G)= \min_{C}\, \mu(C)$ the {\em minimum mean length} of a cycle 
in $G$, and we say that $C^*$ is a {\em minimum mean cycle} 
if ${C^*\in\mathrm{argmin}_{C}\, \mu(C)}$. 
We can find a minimum mean cycle in polynomial time using the classical techniques of Karp~\cite{Kar78}. 
\begin{theorem}\label{thm:mu}
	An $\epsilon$-approximate valuation function which (approximately) 
	rationalises bidding behaviour exists if and only if the bidding graph 
	has minimum mean cycle $\mu(G)\geq -\epsilon$.
\end{theorem}
\begin{proof} 
	From the bidding graph $G$ we create an auxiliary directed graph 
	$\hat{G}=(\hat V, \hat A)$ with vertex set 
	$\hat V = \{\bx_1, \bx_2,\dots, \bx_T\}$. The arc set is complete with arc
	lengths \[
		\hat{\ell}_{\bx_s,\bx_t} = \ell_{\bx_s, \bx_t} - \mu(G) 
		\enspace. 
	\] Observe that, by construction, every cycle in $\hat{G}$ is of 
	non-negative length. 
	It follows that we may obtain shortest path distances $\hat{d}$ from 
	any arbitrary root vertex $r$. 
	Thus, for any arc $(\bx_t, \by)$, we have 
	\begin{align*}
		\hat{d}(\by)
		& \,\leq \, \hat{d}(\bx_t)+ \hat{\ell}_{\bx_t,\by} \\
		& \, = \,   \hat{d}(\bx_t)+ \ell_{\bx_t,\by} - \mu(G) \\
		& \,\leq \, \hat{d}(\bx_t)+ \bp_t\cdot(\by-\bx_t) - \mu(G) 
		\enspace.
	\end{align*} 
	So, if we set $\hat{v}(\bx)=\hat{d}(\bx)$, for each $\bx$, then \[
		\hat{v}(\bx_t)- \bp_t \cdot\bx_t 
		\,\ge\, \hat{v}(\by) - \bp_t \cdot \by +\mu(G)
		\enspace .
	\] for all $t$. 
	Therefore, by definition of $\epsilon$-approximate bidding, 
	we have that $\hat v$ is a $(-\mu)$-approximate virtual valuation function.

	Conversely, let $\hat v$ be an $\epsilon$-approximate virtual valuation 
	function which rationalises the graph, and take some cycle $C$ of minimum 
	mean length in the bidding graph. 
	Suppose for a contradiction that $\mu(C)< -\epsilon$. 
	By \mbox{$\epsilon$-approximability}, we have \[
		\hat v(\bx_s) -\bp_s\cdot\bx_s 
		\, \ge \, 
		\hat v(\bx_t) -\bp_s\cdot\bx_t -\epsilon
		\enspace . 
	\] But $\ell_{\bx_s,\bx_t}\geq\bp_s\cdot (\bx_t-\bx_s)$. 
	Therefore $\ell_{\bx_s\bx_t} \, \geq\, \hat v(\bx_t)-\hat v(\bx_s) -\epsilon$. 	
	Summing over every arc in the cycle we obtain \[
		\ell(C) 
		\, =\,  
		\sum_{(\bx,\by)\in C}\ell_{\bx\by}  
		\, \geq\,  
		\sum_{(\bx,\by)\in C}\left(
			\hat v(\by)-\hat v(\bx) -\epsilon
		\right) 
		\, =\,  
		-|C|\cdot \epsilon
		\enspace .
	\] Thus $\mu(C) \ge -\epsilon$, giving the desired contradiction. \qed
\end{proof}

Recall that, the bidding behaviour is irrational only if $\mu(G)$ is 
strictly negative. 
We emphasize that Theorem \ref{thm:mu} applies even when $\mu(G)$ is positive, 
but in this case, we have an $\epsilon$-approximate virtual valuation function 
where $\epsilon$ is negative! 
What does this mean? 
Well, setting $\delta=-\epsilon$, we then have, for all $t$ and 
for any bundle $\by$, that 
$\hat v(\bx_t) -\bp_t\cdot\bx_t \, \ge\,  \hat v(\by)- \bp_t\cdot \by +\delta$.
Thus, $\bx_t$ is not just the best choice, 
but it provides at least an extra $\delta$ units of utility over any other bundle.
Thus, the larger $\delta$ is, the greater our degree of confidence in the 
revealed preference-ordering and valuation.

\subsection{Individually Rational Virtual Valuation Functions}
Theorem \ref{thm:mu} shows how to obtain a virtual valuation function with the 
best possible additive approximation guarantee: 
any valuation rationalising the bidding graph $G$ must allow for an additive 
approximation of at least $-\mu(G)$. 
However, there is a problem. 
Such a valuation function may not actually be compatible with the data; 
specifically, it may not be individually rational. 
For {\em individual rationality}, we require, for each time $t$, 
that $\hat{v}(\bx_t)-\bp_t\cdot \bx_t \geq 0$. 
But individually rationality is (almost certainly) violated for the the 
root node~$r$ since we have $\hat{v}(\bx_r)=0$. 

It is possible to obtain an individually rational, approximate, 
virtual valuation function simply by taking the $\hat v$ from Theorem~\ref{thm:mu}  
and adding a huge constant to value of each package. 
This operation, of course, is entirely unnatural and the resulting 
valuation function is of little practical value.

\subsubsection{The Minimum Individually Rational Virtual Valuation Function.} \label{sub:min}
We say that $v()$ is the {\em minimum individually rational, 
$\epsilon$-approximate virtual valuation function} if 
$v(\bx_t)\le \omega(\bx_t)$ for each $1\le t\le T$, for any other individually 
rational, $\epsilon$-approximate virtual valuation function $\omega()$. 
This leads to the questions: 
(i) Does such a valuation function exist? and 
(ii) Can it be obtained efficiently? 
The answer to both these questions is $yes$.
\begin{theorem}\label{thm:min-mu}
	The minimum individually rational, $\mu$-approximate virtual 
	valuation function exists and can be found in polynomial time.
\end{theorem}
\begin{proof}
	We create an auxiliary directed graph $H$ from $\hat G$ by adding a sink 
	vertex $\bz$.
	We add an arc $(\bx_t, \bz)$ of length $-\bp_t\cdot \bx_t$, 
	for each $1\le t\le T$, allowing for repeated arcs. 
	Because $\hat{G}$ contains no negative cycle, 
	neither does $H$. 
	Therefore, there exist shortest path distances in $H$. 
	Denote by $\hat{d}()$ the shortest path distance from vertex $\bx_t$ 
	to $\bz$ in $H$. 
	We claim that setting $v(\bx_t)=-\hat{d}(\bx_t)$ gives the minimum 
	individually rational, $\mu$-approximate virtual valuation function. 

	To begin, let's verify that $v()$ is an individually rational, 
	$\mu$-approximate virtual valuation function. 
	First, we require that $v()$ is individually rational. 
	Now the direct path consisting of the arc $(\bx_t, \bz)$ is at least as 
	long as the shortest path from $\bx_t$ to $\bz$. 
	Thus, $- \bp_{t}\cdot\bx_{t} \geq \hat{d}(\bx_t) $. 
	Individual rationality then follows as 
	$v(\bx_t)= -\hat{d}(\bx_t) \geq \bp_{t}\cdot\bx_{t}$. 
	
	Second we need to show that $v()$ is $\mu$-approximate. 
	Consider a pair $\{ \bx_s, \bx_t\}$. 
	The shortest path conditions imply that \[
		-v(\bx_s) 
		\, = \, 
		\hat d(\bx_s)
		\, \le \, 
		\hat{\ell}_{st} + \hat{d}(\bx_t)
		\, = \,  
		(\ell_{st}-\mu) + \hat{d}(\bx_t)
		\, = \, 
		(\ell_{st}-\mu) -v(\bx_t)
		\enspace .
	\] Here the inequality follows from the shortest path conditions on 
	$\hat d()$. 
	Therefore, by definition of $\ell_{st}$, 
	\begin{align*}
		v(\bx_t) 
		&\,\leq\, v(\bx_s) + \ell_{st}-\mu \\
		&\,=\,    v(\bx_s) + \min_{s'} \,\{ 
			\bp_{s'}\cdot (\bx_t - \bx_s): \bx_{s'}=\bx_s
		\} -\mu \\
		&\,\leq\, v(\bx_s) + \bp_{s}\cdot(\bx_{t}-\bx_s) -\mu
		\enspace .
	\end{align*}
	Hence, $v()$ is $\mu$-approximate as desired.

	Finally we require that $v()$ is minimum individually rational. 
	So, take any other individually rational, 
	$\mu$-approximate virtual valuation $\omega()$. 
	We must show that $v(\bx_t)\le \omega(\bx_t)$ for every bundle $\bx_t$. 
	Now consider the shortest path tree $T$ in $H$ corresponding to $\hat d()$. 
	If $(\bx_t, \bz)$ is an arc in $T$ (and at least one such arc exists) then 
	$- \bp_{t}\cdot\bx_{t} = \hat{d}(\bx_t) $. Thus \[
		v(\bx_{t}) - \bp_{t}\cdot\bx_{t}
		\, = \, 
		(- \bp_{t}\cdot\bx_{t})-\hat{d}(\bx_{t})
		\, = \,	0 
		\, \le \,   
		\omega(\bx_{t}) - \bp_{t}\cdot\bx_{t}
		\enspace .
	\] Here the inequality follows by the individual rationality of $\omega()$. 
	Thus ${v(\bx_{t})\, \leq\, \omega(\bx_{t})}$. 
	Now suppose that $v(\bx_s) > \omega(\bx_s)$ for some $\bx_s$. 
	We may take $\bx_s$ to be the closest vertex to the root $\bz$ in $T$ with 
	this property. 
	We have seen that $\bx_s$ cannot be a child of $\bz$. 
	So let $(\bx_s, \bx_t)$ be an arc in $T$. 
	As $\bx_t$ is closer to the root than $\bx_s$, we know 
	$v(\bx_t) \le \omega(\bx_t)$. 
	Then, as $T$ is a shortest path tree, 
	we have $\hat d(\bx_s) = \hat{\ell}_{st} + \hat{d}(\bx_t)$. 
	Consequently $-v(\bx_s) = \hat{\ell}_{st} -v(\bx_t)$, 
	and so 
	\[
		\omega(\bx_t)
		\, \geq \, 
		v(\bx_t) 
		\, = \, 
		\hat{\ell}_{st} +  v(\bx_s)
		\, > \, \hat{\ell}_{st} +  \omega(\bx_s)
		\enspace .
	\] But then \[ 
		\omega(\bx_t) 
		\, > \, 
		\omega(\bx_s) +  \ell_{st} -\mu 
		\, = \,
		\omega(\bx_s) + \min_{s'} \,\{
			\bp_{s'}\cdot (\bx_t - \bx_s): \bx_{s'}=\bx_s
		\} -\mu 
		\enspace .
	\] 
	
	It follows that there is at least one time period when $\bx_s$ was selected 
	in violation of the $\mu$-optimality of $\omega()$. 
	So $v()$ is a minimum individually rational, 
	$\mu$-approximate virtual valuation function. \qed
\end{proof}

\subsubsection{The Maximum (Individually Rational) Virtual Valuation Function.}\label{sub:max}
The minimum individually rational virtual valuation function allows 
us to obtain worst-case social welfare guarantees when revealed preference 
is used in mechanism design, see Section~\ref{sec:mech-design}. 
For the best-case welfare guarantees, we are interested in finding the 
{\em maximum} virtual valuation function. 
In general, this need not exist as we may add an arbitrary constant to 
each bundle's valuation given by the minimum individually rational virtual 
valuation function. 
But, it does exist provided we have an upper bound on the valuation of at 
least one bundle. 
This is often the case. 
For example in a combinatorial auction if a bidder drops out of the auction 
at time $t+1$, then $\bp_{t+1}\cdot \bx_t$ is an upper bound on the value of 
bundle $\bx_t$. 
Furthermore, in practice, bidders (and the auctioneer) often have 
(over)-estimates of the maximum possible value of some bundles. 

So suppose we are given a set $I$ and constraints of the form 
$v(\bx_i)\leq \beta_i$ for each $i\in I$. 
Then there is a {\em unique} maximum $\mu$-approximate virtual valuation function.

\begin{theorem}\label{thm:max-mu}
Given a set of constraints, the maximum $\mu$-approximate virtual valuation 
function exists and can be found in polynomial time. 
\end{theorem}
\begin{proof}
	Let $v(\bx_i)\leq \beta_i$ for each $i\in I$. 
	We construct a graph $H$ from $\hat G$ by adding a source vertex $\bz$ 
	with arcs of length $\beta_i$ from $\bz$ to $\bx_i$, for each $i\in I$. 
	Since $\bz$ has in-degree zero, $H$ has no negative cycles because 
	$\hat{G}$ does not. 
	Denote by $\hat d()$ the shortest distance of every vertex {\em from} $\bz$. 
	We claim that setting $v(\bx) = \hat d(\bx)$ gives us the desired maximum 
	$\mu$-approximate valuation function.
		
	To prove this, we first begin by checking that it satisfies the 
	upper-bound constraints. 
	This is trivial, because for each $i\in I$ there is a path consisting of 
	one arc 	of length $\beta_i$ from $\bz$ to $\bx_i$. 
	Thus the shortest path to $\bx_i$ has length at most $\beta_i$.
	Second, the valuation function $v() = \hat d()$ is $\mu$-approximate by 
	the 	choice of arc length in $\hat{G}$.
	Third, we show that this valuation function is maximum. 
	So, take any other $\mu$-approximate virtual valuation $\omega()$ that 
	satisfies the upper bound constraints $I$. 
	We must show that $v(\bx_t)\ge \omega(\bx_t)$ for every bundle $\bx_t$.
	For a contradiction, suppose that $P=\{\bz, \by_1, \by_2, \dots,  \by_r\}$ 
	is the shortest path from $\bz$ to $\by_r$ in $H$ and that 
	$v(\by_r)< \omega(\by_r)$. 
	Observe that the node adjacent to $\bz$ on $P$ must be $\by_1=\bx_i$ for 
	some $i\in I$.
	Now because $\omega()$ is a $\mu$-approximate valuation function, we have
	\begin{equation*}
		\sum_{j=1}^{r-1} \omega(\by_{j+1}) 
		\, \leq\,  
		\sum_{j=1}^{r-1} \left(
			\omega(\by_j) + \ell_{\by_j,\by_{j+1}}-\mu
		\right) 
		\, =\,
		\sum_{j=1}^{r-1} \left(
			\omega(\by_j) + \hat{\ell}_{\by_j,\by_{j+1}}
		\right)
		\enspace .
	\end{equation*}
	Cancelling terms produces
	\begin{equation*}
	\omega(\by_{r}) 
		\, \leq\,  
		\omega(\by_1) +\sum_{j=1}^{r-1} \hat{\ell}_{\by_j,\by_{j+1}}
		\, \leq\,  
		\beta_j +\sum_{j=1}^{r-1} \hat{\ell}_{\by_j,\by_{j+1}}
		\, =\,  
		\hat{d}(\by_r)
		\, =\, 
		v(\by_r)
		\enspace .
	\end{equation*}
	Here the second inequality follows by the facts that 
	$\by_1=\bx_i$, for some $i\in I$, and $\omega()$ satisfies the 
	upper bound constraint $\omega(\bx_i) \le \beta_i$.
	This contradicts the assumption that $v(\by_r)< \omega(\by_r)$.
	\qed
\end{proof}

Notice that Theorem \ref{thm:max-mu} does not guarantee that the maximum 
virtual valuation function is individually rational. 
For example, suppose $\beta_t=\bp_t\cdot \bx_t$, for all $1\le t\le T$. 
Individual rationality then implies that $v(\bx_t)$ must equal 
$\bp_t\cdot \bx_t$ for every bundle.
In general, however, such a valuation function is not $\mu$-approximate. 
In such cases no individually rational $\mu$-approximate virtual valuation 
functions may exist that satisfy the upper bound constraints. 
On the other hand, suppose such a virtual valuation function does exist. 
Then the maximum $\mu$-approximate virtual valuation function in 
Theorem~\ref{thm:max-mu} must be individually rational by maximality.

\section{Revealed Preference Auction Bidding Rules}\label{sec:mech-design}
So far, we have focused upon how to test the degree of rationality reflected 
in a data set. 
Specifically, we saw in Theorem \ref{thm:mu} that the minimum mean length 
of a cycle, $\mu(G)$, gives an exact and optimal goodness of fit measure 
for rationality. 
Furthermore, Theorem \ref{thm:min-mu} explained how to quickly obtain the 
minimum individually rational valuation function that best fits the data.

Recall, however, that revealed preference is also used as tool in mechanism design.
In particular, we saw in Section~\ref{sec:auction-rp} how revealed preference 
is used to impose bidding constraints in combinatorial auctions. 
We will now show how to apply the combinatorial arguments we have developed 
to create other relaxed revealed preference constraints.

\subsection{Relaxed Revealed Preference Bidding Rules}
Consider a combinatorial auction at time (round) $t$ where our prior 
price-bundle bidding pairs are 
$\{(\bp_1, \bx_1), (\bp_2, \bx_2),  \dots, (\bp_{t-1}, \bx_{t-1})\}$. 
By Inequality (\ref{eq:RP-auction}) in section \ref{sec:auction-rp}, 
rational bidding at time $t$ implies that \[
	v(\bx_t) -\bp_t\cdot\bx_t 
	\, \geq \, 
	v(\bx_s) -\bp_t\cdot\bx_s,
	\quad \text{for all $s< t$.}
\] Moreover, a necessary condition is then that 
${(\bp_t-\bp_s)\cdot\bx_s \,\geq \, (\bp_t-\bp_s)\cdot\bx_t}$ and this can 
easily be checked by searching for negative length digons in the 
bidding graph induced by the first $t$ bids. 
If such a cycle is found then the bid $(\bp_t, \bx_t)$ is not permitted 
by the auction mechanism.

The non-permittal of bids is clearly an extreme measure, and one that can 
lead to the exclusion of bidders from the auction even when they still have 
bids they wish to make. 
In this respect, it may be desirable for the mechanism to use a relaxed set 
of revealed preference bidding rules. 
The natural approach is to insist not upon strictly rational bidders but rather 
just upon approximately rational bidders. 
Specifically, the auction mechanism may (dynamically) select a desired 
degree $\epsilon$ of rationality. 
This requires that at time $t$, \[
	v(\bx_t) -\bp_t\cdot\bx_t 
	\, \geq \, 
	v(\bx_s) -\bp_t\cdot\bx_s -\epsilon,
	\quad \text{for all $s< t$.} 
\] A necessary condition then is 
$
	(\bp_t-\bp_s)\cdot\bx_s  
	\, \ge\, 
	(\bp_t-\bp_s)\cdot\bx_t -2\epsilon
$, and we can test this {\em relaxed} {\sc warp}-based bidding rule by 
insisting that every digon has mean length at least $-\epsilon$. 
Similarly, the {\em relaxed} {\sc karp}-based bidding rule is 
\begin{equation}\label{eq:relaxed-karp}
	(\bp_k-\bp_0)\cdot\bx_0 
	\, \ge\,  
	\sum_{i=1}^k  (\bp_i-\bp_{i-1})\cdot\bx_i - (k+1)\cdot \epsilon
\end{equation}

The {\em relaxed} {\sc garp}-based bidding rule applies the relaxed 
{\sc karp}-based bidding rule for every choice of $k$. 
The imposition of the relaxed {\sc garp}-based bidding rule ensures 
approximate rationality.

\begin{theorem}\label{thm:relaxed-garp}
	A set of price-bid pairings $\{(\bp_t, \bx_t): 1\le t\le T\}$ 
	has a corresponding $\epsilon$-approximate individually rational 
	virtual valuation function if and only if it satisfies the relaxed 
	{\sc garp}-based bidding rule. 
\end{theorem}
\begin{proof}
	Suppose the relaxed {\sc garp}-based bidding rule is satisfied.
	By Theorem~\ref{thm:mu}, it suffices to show that the minimum mean 
	cycle in the bidding graph with arc lengths $\ell$ is at least $-\epsilon$.
	So take any collection $\{\bx_i\}_{i=1}^k$ of bundles. 
	Let $t_i$ be the time when $\ell_{\bx_i,\bx_{i+1}}$ was minimized, 
	and let $\bp_i:=\bp_{t_i}$.
	Then we have 
	\begin{align*}
		-(k+1)\cdot \epsilon 
	   	&\, \le\,  
		(\bp_k-\bp_0)\cdot\bx_0 - 
			\sum\nolimits_{i=1}^k  (\bp_i-\bp_{i-1})\cdot\bx_i\\
		&\, =\,  
		\sum\nolimits_{i=0}^k  \bp_i\cdot (\bx_{i+1}-\bx_{i}) \\
		&\, =\,  
	    \sum\nolimits_{i=0}^k  \ell_{\bx_i,\bx_{i+1}}
     \end{align*}
	Here, the inequality follows because the relaxed {\sc garp}-based 
	bidding rule is satisfied. 
	(Again the subscripts are taken modulo $k+1$.) 
	Since, the corresponding cycle contains $k+1$ arcs,
	we see that the length of the minimum mean cycle is at least $-\epsilon$.

	Conversely, if the bidding data has a corresponding 
	$\epsilon$-approximate individually rational virtual valuation function 
	then the relaxed bidding rules are satisfied.
	\qed
\end{proof}
 
\subsection{Relaxed KARP-Based Bidding Rules}
Theorem \ref{thm:relaxed-garp} tells us that imposing the relaxed 
{\sc garp}-based bidding rule ensures approximate rationality. 
But, in practice, even {\sc warp}-based bidding rules are often confusing to 
real bidders. 
There is likely therefore to be some resistance to the idea of imposing the 
whole gamut of {\sc garp}-based bidding rules. 
We believe that this combinatorial view of revealed preference, 
where the bidding rules can be tested via cycle examination, 
will eradicate some of the confusion. 
However, for simplicity, there is some worth in quantitatively examining 
the consequences of imposing a weaker relaxed {\sc karp}-based bidding rule 
rather than the {\sc garp}-based bidding rule. 
To test for the relaxed {\sc karp}-based bidding rules, we simply have to 
examine cycles of length at most $k+1$. 
Now suppose the {\sc karp}-based bidding rules are satisfied.
By finding the $\mu(G)$ in the bidding graph we can still obtain the 
best-fit additive approximation guarantee, 
but we no longer have that this guarantee is $\epsilon$. 
We can still, though, prove a strong additive approximation guarantee even 
for small values of $k$. To do this we need the following result. 
\begin{theorem}\label{thm:kcycle}
	Given a complete directed graph $G$ with arc lengths $\ell$. 
	If every cycle of cardinality at most $k+1$ has non-negative length 
	then the minimum mean length of a cycle is at least $-\frac{\ell^{\max}}{k}$,
	where $\ell^{\max}=\max_{e\in E(G)} |\ell_e|$.
\end{theorem}
\begin{proof}
	Take any cycle $C$ with cardinality $|C|>k+1$.
	Let the arcs of $C$ be $\{e_1,\,e_2,\,\dotsc,\, e_{|C|}\}$ in order.
	Then
	\begin{equation}\label{eq:kcycleA}
		\sum_{i=1}^{|C|} \,\sum_{j=i}^{i+k-1} \ell_{e_j} 
		\, =\,  
		k \cdot \sum_{i=1}^{|C|} \ell_{e_i}
		\, =\, 
		k\cdot \ell(C)
		\, =\,  
		k\cdot |C| \cdot \frac{\ell(C)}{|C|}
		\enspace .
	\end{equation}
	Above, the inner summation is taken modulo $|C|$.
	On the other hand take any path segment 
	$P=\{e_i, e_{i+1},\dots, e_{i+k-1}\}$, where again the subscript summation 
	is modulo $|C|$. 
	Because the graph is complete and the maximum arc length is $\ell^{\max}$, 
	the length of $P$ is at least $-\ell^{\max}$. 
	Otherwise, we have a negative length cycle of cardinality $k+1$ by adding 
	to $P$ the arc from the head vertex of $e_{i+k-1}$ to the tail vertex of 
	$e_i$. Thus, 
	\begin{equation}\label{eq:kcycleB}
		\sum_{i=1}^{|C|} \,\sum_{j=i}^{i+k-1} \ell_{e_j} 
		\, \ge\,  
		-|C|\cdot \ell^{\max} \enspace .
	\end{equation}
	Combining Equalities (\ref{eq:kcycleA}) and Inequality (\ref{eq:kcycleB}) 
	gives that ${\frac{\ell(C)}{|C|} \, \ge\, -\frac{\ell^{\max}}{k}}$.
	As every cycle of cardinality at most $k+1$ has non-negative mean length, 
	this implies that the minimum mean length 
	of any cycle in $G$ is at least $-\frac{\ell^{\max}}{k}$.
	\qed
\end{proof}

This result is important as it allows us to bound the degree of 
rationality that must arise whenever we impose the relaxed 
{\sc karp}-based bidding rule.
\begin{corollary}\label{cor:kcycle}
	Given a set of price-bid pairings ${\{(\bp_t, \bx_t): 1\le t\le T\}}$ 
	that satisfy the  relaxed {\sc karp}-based bidding rule,
	there is a ${(\frac{b^{\max}}{k}+\epsilon)}$-approximate individually 
	rational virtual valuation function, 
	where $b^{\max}$ is the maximum bid made by the bidder during the 
	auction. 
\end{corollary} 
\begin{proof} 
	The relaxed {\sc karp}-based bidding rule (\ref{eq:relaxed-karp}) 
	implies that every cycle of cardinality at most $k+1$ in the bidding 
	graph $G$ has mean length at least $-\epsilon$.
	Let $G'$ be the modified graph with arc lengths 
	${\ell'_{\bx_s,\bx_t}:=  \ell_{\bx_s,\bx_t}+\epsilon}$.
	Then every cycle in $G'$ of cardinality at most $k+1$ has
	non-negative length. 
	By Theorem \ref{thm:kcycle}, the minimum mean length of a cycle in $G'$ 
	is then at most $\frac{(\ell')^{\max}}{k}$. 
	Furthermore, $(\ell')^{max} = \ell^{max} +\epsilon \le b^{\max}+\epsilon$. 
	Theorems~\ref{thm:mu} and \ref{thm:min-mu} then guarantee the existence 
	of  a $(\frac{b^{\max}}{k}+\epsilon)$-approximate individually 
	rational virtual valuation function. \qed
\end{proof}

One may ask whether the additive approximation guarantee in
Corollary~\ref{cor:kcycle} can be improved. The answer is {\em no}; Theorem~\ref{thm:kcycle} is tight.

\begin{lemma}\label{lem:kcycle-tight}
	There is a graph $G$ where each cycle of cardinality at most $k+1$ 
	has non-negative length and the minimum mean length of a cycle is 
	$-\ell^{\max}/k$.
\end{lemma}
\begin{proof}
	Let $G$ be a complete directed graph with vertex set 
	${V=\{v_1,v_2,\dots,v_{n}\}}$.
	We will define arc lengths $\ell$ such that all $(k+1)$-cycles in 
	$G$ have non-negative length, 
	but the minimum mean length of a cycle is $-\frac{\ell^{\max}}{k}$. 
	First consider the cycle $C_0=\{v_1,v_2,\dots,v_{k+2}, v_1\}$. 
	Give each arc in $C_0$ a length $-\frac{\ell^{\max}}{k}$.
	Thus $C_0$ has cardinality $k+2$ and mean length $-\frac{\ell^{\max}}{k}$. 
	Now let every other arc $e$ have length $\ell^{\max}$.
	It immediately follows that the only cycle in $G$ with negative length is 
	$C_0$. 
	Thus, all cycles of length at most $k+1$ have non-negative length,
	but the minimum mean length of a cycle is $-\frac{\ell^{\max}}{k}$, 
	as desired.
\end{proof}

\subsection{Welfare Guarantees with Revealed Preference Rules}
Our results from Section \ref{sec:approx-vv} give approximate rationality 
guarantees on individual bidders.
We briefly outline this here. 
By applying the above relaxed revealed preference bidding rules to each bidder, 
we can now obtain guarantees on the overall social welfare of the entire auction. 
For example, consider imposing the relaxed {\sc garp} bidding rules. 
Now suppose each bidder uses a minimum, individually rational, 
$\epsilon$-approximate virtual valuation function that satisfies the 
gross substitutes property. 
It is known that if bidder valuation functions satisfy the gross 
substitutes property then the simultaneous multi-round auction 
(SMRA) will converge to a Walrasian equilibrium and maximize social 
welfare~\cite{Mil00,GS99,KC82}.
Consequently, the output allocation now maximizes virtual welfare to within
an additive factor per bidder. 
One may expect there is some maximum discrepancy (say, in the $L^{\infty}$ norm) 
between the true valuation function and {\em some} virtual approximate 
virtual valuation. 
If so, because the implemented virtual valuations are minimum, 
we can then lower-bound the true social welfare. 
Similarly, best-case bounds follow using the maximum approximate virtual 
valuation function.

\subsection{Alternate Bidding Rules}\label{sec:other}
Interestingly other bidding rules used in practice or proposed in the literature 
can be viewed in the graphical framework. 
For example, bid withdrawals correspond to vertex deletion in the bidding graph,
whilst budget constraints and the Afriat Efficiency Index can be formulated in 
terms of arc-deletion. We briefly describe these applications here.


\subsubsection{Revealed Preference with Budgets.}
Recall that, in Section \ref{sec:auction-rp}, we have assumed that, in the 
quasilinear model, bidders have no budgetary constraints. 
This is not a natural assumption. 
Harsha et al.~\cite{HBP10} explain how to implement budgeted revealed 
preference in a combinatorial auction. 
Their method applies to the case when the fixed budget $B$ is unknown to the 
auction mechanism. 
To do this, upper and lower bounds on feasible budgets are maintained 
dynamically via a linear program. 
It is also straightforward to do this combinatorially using edge-deletion in
the bidding graph; we omit the details as the process resembles that 
of the following subsection.

\subsubsection{The Afriat Efficiency Index.}
Recall that to determine the Afriat Efficiency Index we reveal 
$\bx_t\succeq \by$ only if $\bp_t\cdot\by \le \lambda\cdot \bp_t\cdot\bx_t$ 
where $\lambda<1$. 
This is equivalent, in Afriat's original setting, to removing from the graph 
any arc $(\bx_t, \bx_s)$ for which $\bp_t\cdot\bx_s > \lambda\cdot \bp_t\cdot\bx_t$. 
Of course, for the application of combinatorial auctions, we assume 
quasi-linear utilities. 
Therefore, the appropriate implementation is to remove any arc 
$(\bx_t, \bx_s)$ for which 
\[
	v(\bx_s)-\bp_t\bx_s > \lambda \cdot (v(\bx_t)-\bp_t\bx_t)
	\enspace .
\]
How, though, can we implement this rule as $v()$ is unknown? 
We can simply apply the techniques of Section \ref{sec:approx-vv} and use for 
$v$ the minimum individually rational virtual valuation function.
We can now determine the best choice of $\lambda$ that gives a 
predetermined, $\epsilon$ additive approximation guarantee $\epsilon$. 
This can easily be computed exactly by bisection search over the set of arcs, 
as each arc $a$ has its own critical value $\lambda_a$ at which it will be removed.
The optimal choice arises at the point where the minimum mean cycle in the 
bidding graph rises above $-\epsilon$.  
When $\epsilon=0$, the corresponding choice of $\lambda$ is the anolog
of the Afriat Efficiency Index.

\subsubsection{Revealed Preference with Bid Withdrawals.}
Some iterative multi-item auctions allow for bid withdrawals, most notably 
the simultaneous multi-round auction (SMRA). 
Bid withdrawals may easily be implemented along with revealed preference bidding rules. At time $t$, a bid withdrawal
corresponds to the removal of (a copy of) a vertex $\bx_s$, where $s<t$.
This may be important strategically. To see this, suppose the bid $\bx_t$
is invalid under the {\sc karp}-based bidding rules because it
would induce a negative cycle of cardinality at most $k+1$ in the bidding graph on $\{\bx_1,\bx_2,\dots,\bx_t\}$.
If $\bx_s$ lies on all such negative cycles then $\bx_t$ becomes a valid bid
after the withdrawal of $\bx_s$. Because auctions typically restrict the total number of bid withdrawals 
allowed, the optimal application of bid withdrawals correspond to the problem of
finding small hitting sets for the negative length cycles of cardinality at most $k+1$.

 \end{document}